\documentclass[12pt]{article}

\usepackage{enumerate}
\usepackage{hyperref}
\usepackage{graphicx}
\usepackage{amsmath, amsthm}
\usepackage{amsmath,amssymb}
\usepackage{natbib}
\usepackage{setspace}

\allowdisplaybreaks

\theoremstyle{definition}
\newtheorem{remark}{Remark}
\newtheorem{lemma}{Lemma}

\theoremstyle{plain}
\newtheorem{theorem}{Theorem}
\newtheorem{assum}{Assumption}

\newcommand{\cov}{\mathbf{Cov}}

\newcommand{\Var}{\mathbf{Var}}

\newcommand{\E}{\mathbf{E}}
\newcommand{\X}{\mathcal{X}}

\newcommand\smallO{
	\mathchoice
	{{\scriptstyle\mathcal{O}}}
	{{\scriptstyle\mathcal{O}}}
	{{\scriptscriptstyle\mathcal{O}}}
	{\scalebox{.5}{$\scriptscriptstyle\mathcal{O}$}}
}

\newcommand\numberthis{\addtocounter{equation}{1}\tag{\theequation}}

\title{Bibliography management: \texttt{biblatex} package}
\usepackage[ruled,vlined]{algorithm2e} 

\title{Multivariate strong invariance principles in \\ Markov chain Monte Carlo}
\newcommand{\footremember}[2]{%
	\footnote{#2}
	\newcounter{#1}
	\setcounter{#1}{\value{footnote}}%
}

\author{%
	Arka Banerjee\footremember{alley}{Department of Mathematics and Statistics, Indian Institute of Technology, Kanpur, India 208016. Email: \href{arkabee20@iitk.ac.in}{arkabee20@iitk.ac.in} .}%
	\and Dootika Vats\footremember{alle2}{Department of Mathematics and Statistics, Indian Institute of Technology, Kanpur, India 208016. Email: \href{dootika@iitk.ac.in}{dootika@iitk.ac.in} .}%
}

\begin{document}
	\onehalfspacing
	\maketitle
	
	\begin{abstract}
		Strong invariance principles in Markov chain Monte Carlo are crucial to theoretically grounded output analysis. Using the wide-sense regenerative nature of the process, we obtain explicit bounds on the  almost sure convergence rates for partial sums of multivariate ergodic Markov chains.  Further, we present results on the existence of strong invariance principles for both polynomially and geometrically ergodic Markov chains without requiring a $1$-step minorization condition. Our tight and explicit rates have a direct impact on output analysis, as it allows the verification of important conditions in the strong consistency of certain variance estimators.
	\end{abstract}

	\section{Introduction}
	\label{sec:intro}
	
	Markov chain Monte Carlo (MCMC) is the workhorse computational algorithm for Bayesian inference. In MCMC, given an intractable target distribution -- typically a multi-dimensional Bayesian posterior -- an ergodic Markov chain is constructed such that its stationary distribution is the desired target distribution. The increasing complexity of modern data and modeling strategies  warrants the need for ensuring quality assessment of MCMC output. Consistent estimation of the variance of MCMC estimators often requires the assumption of a strong invariance principle (SIP) on the underlying Markov chain \citep{chan2022mean,pengel2021strong,vats2018strong}. Utilizing wide-sense regenerative properties of ergodic Markov chains, we present conditions for when a multivariate strong invariance principle holds under polynomial and geometric ergodicity. Our rates are explicit and the tightest known in this framework. We further explain how the explicit rates have a direct consequence on practical variance estimation in MCMC.

	Let $\{X_{t}\}_{t\ge1}$ be a $d$-dimensional, $\pi$-stationary stochastic process on a measurable space $(\mathcal{X}, \mathcal{B(X)})$ and denote the sequence of partial sum with $S_{n} = \sum_{t=1}^{n} X_{t}$. Let $\Vert \cdot \Vert$ denote the Euclidean norm, $L$ be a $d \times d$ positive-definite matrix, and $\kappa:\mathbb{N} \rightarrow \mathbb{R^{+}}$ be an increasing function.  A multivariate SIP holds if on a suitably rich probability space one can construct $\{X_{t}\}_{t\ge1}$ along with a $d$-dimensional Wiener process $\{W(t): t \ge 0\}$ so that as $n \to \infty$,
	\begin{equation*}
		\left\Vert S_{n} - \E(S_{n}) - L W(n) \right\Vert = \mathcal{O}(\kappa(n)) \qquad \text{with probability 1}\,.
	\end{equation*}
	
	The rate $\kappa(n)$ often depends on the moments and the amount of correlation in the process. For independent and identically distributed (iid) univariate processes exhibiting a moment generating function, \cite{komlos1975approximation,komlos1976approximation} (KMT) obtained the rate $\kappa(n) = \log(n)$;  if $X_1$ has $r$ moments for $r > 3$, they obtain rate $\kappa(n) = n^{1/r}$. These are the tightest rates possible. 	For correlated sequences, \cite{stout1975almost} collate the various known SIP rates  for $\phi$-mixing, $\alpha$-mixing, non-stationary, and regenerative processes. The rates obtained are of the form $\kappa(n) = n^{1/2 - \lambda}$ where $\lambda$ is known.
	
	The situation is quite different in the multivariate case. Although, for iid vectors \cite{einmahl1989extensions} extend the KMT result, the proof techniques used in the univariate case do not, in general, yield explicit rates in the multidimensional one \citep{monrad1991problem}. For general $\phi$-mixing processes,
	\cite{berkes1979approximation,eberlein1986strong,kuelbs1980almost} obtain a rate of $n^{1/2 - \lambda}$ for some unknown $0 < \lambda < 1/2$. Making certain parametric assumptions on the process and building on the work of \cite{wu2007strong}, \cite{liu2009strong,berkes2014komlos} obtain explicit near-optimal and optimal rates; these assumptions are appropriate in applications such as time-series, but cannot be made for the Markov chains employed in MCMC. Markov chains have thus been given special focus.
	
	For the univariate case, \cite{jones2006fixed} obtain explicit rates for uniformly ergodic Markov chains. For polynomially ergodic univariate  continuous-time Markov processes with $r$ moments, \cite{pengel2021strong} obtain rate $\max\{n^{1/4} \log(n), n^{1/r}\log^2(n)\}$. If the Markov chain exhibits a $1$-step minorization, it is classically regenerative. Let $p$ denote the moments of the regeneration time and let the partial sum of a regenerative tour exhibit $2 + \delta$ finite moments; then \cite{csaki1995additive} obtain  rate $n^{ \max\{1/(2+ \delta), 1/2p, 1/4\}}\log(n)$. This work forms the basis of obtaining explicit rates for Markov chains with \cite{jones2006fixed} obtaining an expression for geometrically ergodic Markov chains and \cite{merlevede2015strong} obtaining the KMT result for geometrically ergodic bounded Markov chains. 
	
	For univariate Markov chains exhibiting an $l$-step minorization, \cite{samur2004regularity} obtain a rate of $n^{1/2 - \lambda}$ for some unknown $0< \lambda < 1/2$.  The best known result under this setting is that of \cite{dong2019new} and \cite{zhu2020asymptotic}, who obtain rate $n^{\max\{1/(2+\delta), 1/p, 1/4\}} \log(n)$ for univariate and multivariate processes, respectively, under moment assumptions on the regeneration time. 
 
    The rates obtained by \cite{dong2019new} and \cite{zhu2020asymptotic} fall short of that of \cite{csaki1995additive}. In Section~\ref{sec:main}, we obtain the \cite{csaki1995additive} rate under a general $l$-step minorization of a multidimensional Markov chain. Further, we show that the moment assumptions required for this result hold for geometrically and polynomially ergodic Markov chains. The tightness and tractability of the bound has a direct impact on consistent estimation of MCMC variances, as we highlight in Section~\ref{sec:practical}. 

    In addition to the almost sure convergence results, $l$-step minorization of a Harris ergodic Markov chain produces a regenerative structure of the underlying Markov chain. This wide-sense regeneration structure yields a regenerative estimator with desirable asymptotic properties. A brief discussion on the regenerative estimators is provided in Section~\ref{sec:reg_est}.
	
	\section{Definitions and main result}
	\label{sec:theory}
	
	\subsection{Markov chains and wide-sense regeneration}
	\label{sec:regen}

	Consider a $\pi$-Harris ergodic Markov chain with a one-step Markov transition kernel
	\[
	P(x, A) := \Pr(X_{k+1} \in A \mid X_{k} = x) \ \ \text{ for } k \ge 1,  x \in \mathcal{X}, \text{   and   } A \in \mathcal{B}(\mathcal{X})\,.
	\] 
	Similarly the $n$-step  Markov transitional kernel is  
	\[
	P^{n}(x, A) := \Pr(X_{k+n} \in A \mid X_{k} = x) \ \ \text{ for } k \ge 1, x \in \mathcal{X}, \text{   and   } A \in \mathcal{B}(\mathcal{X})\,.
	\]
	
	Our SIP results will apply to Markov chains exhibiting certain rates of convergence measured via the total variation distance: 
	\begin{equation}
		\label{eq:tv_dist}
		\Vert P^{n}(x, \cdot) - \pi(\cdot) \Vert_{TV} := \sup_{A \in \mathcal{B}(\mathcal{X})} \left\vert P^{n}(x, A) - \pi(A) \right\vert \leq M(x)G(n)\,,	
	\end{equation}
	where $G(n) \rightarrow 0$ as $n \rightarrow \infty$ and $0< \E_{\pi}[M(X)] < \infty$; here we use notation $\E_F$ to denote expectations when $X_1 \sim F$. If $G(n) = n^{-k}$ for some $k \ge 1$, the Markov chain is polynomially ergodic of order $k$. If $G(n) = t^{n}$ for some $0 \leq t < 1$, the Markov chain is geometrically ergodic. 
	
	Markov chains employed in MCMC are typically $\pi$-Harris ergodic and are thus wide-sense regenerative \citep{glynn2011wide} since they exhibit an $l$-step minorization \citep{athreya1978new}. An $l$-step minorization for $l \ge 1$ holds if there exists  $h: \mathcal{X} \rightarrow [0, 1]$ with $0 < \E_{\pi}[h(X)] < \infty$ and a probability measure $Q(\cdot)$ such that for all $x \in \X, A \in \mathcal{B}(\mathcal{X})$
	\begin{equation}
		\label{eq:mino}
		P^{l}(x, A) \ge h(x) Q(A)\,.
	\end{equation}
	Equation~\eqref{eq:mino} allows the following representation of the $l$-step transition kernel:
	\begin{equation}
		\label{eq:split}
		P^{l}(x, A) = h(x) Q(A) + (1 - h(x))R(x, A), 
	\end{equation}
	where $R(x, \cdot)$ is the residual distribution. By virtue of \eqref{eq:split}, an alternative sampling strategy is possible using an augmented Markov chain, $\{(X^{*}_{t}, \delta_{t})\}_{t \geq 1}$ where $\delta$'s are binary variables. Consider $X^{*}_{1} \sim Q$ and $\delta_{i} \sim  \text{Bernoulli}(h(X^{*}_{i}))$ for $\ i \ge 1$. For $i \ge 2$, if $\delta_{i} = 1$, $X^{*}_{i+l} \sim Q$, else $X^{*}_{i+l} \sim R(X^{*}_{i}, \cdot)$. This alternative sampling strategy is a way of generating a probabilistically similar random process to the original Markov chain with initial kernel $Q(\cdot)$. 
	Every such time-point $i$, so that $\delta_{i} = 1$ is known as a regeneration time. Denote the $k^{\text{th}}$ regeneration time as $T_{k}$, with $T_{0} = 0$. Let $\tau_{k} := T_{k} - T_{k-1}$ be the time to the $k^{\text{th}}$ regeneration from the $(k-1)^{\text{th}}$ regeneration. Denote $\mu := \E_{Q}(\tau_{1})$.
	
	Let $f:\mathcal{X}\rightarrow \mathbb{R}^{d}$ for $d \geq 1$ be a function whose expectation under $\pi$ is of interest. For $k \ge 1$, define the sum of a tour as $Z_{k} := \sum_{t=T_{k-1}+1}^{T_{k}} f(X_{t})$, and denote $\eta := \E_{Q}(Z_{1})$. When $l = 1$, the Markov chain is classically regenerative and $(Z_{k}, \tau_{k})_{k \geq 1}$ are iid
 . For the one-dimensional case, this feature is exploited by \cite{csaki1995additive} to arrive at an SIP using classical KMT results. These results have been adapted to MCMC by \cite{jones2006fixed}. However, for many MCMC samplers, $l = 1$ is a limiting assumption. Since all Harris ergodic Markov chains satisfy an $l$-step minorization for some $l$ \citep[see][for e.g.]{athreya1978new}, the assumption of an $l$-step minorization is no longer limiting. However, for general $l$, \cite{glynn2011wide} discuss that $(Z_{k}, \tau_{k})_{k \geq 1}$ is a one-dependent stationary process, and thus the classical KMT results can no longer be utilized to establish an SIP.

	\subsection{Main result}
	\label{sec:main}
	
	We now present our main results establishing an SIP. Define
	\begin{align*}
		\Sigma_{Z} := &
		\Var_{Q}\left(Z_{1} - \frac{\tau_{1}}{\mu} \eta\right) +  \cov_{Q}\left(Z_{1} - \frac{\tau_{1}}{\mu} \eta, Z_{2} - \frac{\tau_{2}}{\mu} \eta\right)\\
        & + \cov_{Q}\left(Z_{2} - \frac{\tau_{2}}{\mu} \eta, Z_{1} - \frac{\tau_{1}}{\mu} \eta\right).\numberthis \label{eq:sigma_z}
	\end{align*}
	\begin{theorem}
		\label{Ths} Let $\{X_{t}\}_{t\geq 1}$ be a $\pi$-Harris ergodic Markov chain and thus \eqref{eq:mino} holds.  Suppose
		\begin{enumerate}[(a)]
			\item $\E_{Q}(\tau_{1}^{p}) < \infty $ for some $p > 1$  and, 
			\item for some \( \delta > 0 \)   
			\begin{equation}
				\E_{\pi}\left[ \left( \sum_{t=1}^{\tau_{1}} \left\Vert f(X_{t}) - \frac{\eta}{\mu} \right\Vert \right)^{2+\delta}\right] < \infty \,. \numberthis \label{eq:assm-2}
			\end{equation}	
		\end{enumerate}
		Then, on a suitably rich probability space, one can construct $\{X_t\}_{t \geq 1}$ together with a $d$-dimensional standard Wiener process $\{ W(t): t \ge 0 \}$ such that for $\beta = \max\{ 1/{(2+\delta)}, 1/2p, 1/4 \}$, as $n \to \infty$
		\begin{equation}
			\label{eq:main_sip}
			\left\Vert \sum_{t=1}^{n} f(X_{t}) - n \E_{\pi}[f(X)] - \frac{\Sigma_{Z}^{1/2}}{\sqrt{\mu}} W(n) \right\Vert = \mathcal{O}(n^{\beta}\log(n)) \quad \text{with probability 1}\,. 
		\end{equation}
	\end{theorem}
	
	\begin{proof}
		See Appendix~\ref{Thp}.
	\end{proof}
	
	\begin{remark}
		As far as we are aware Theorem~\ref{Ths} is the first result we know that matches the \cite{csaki1995additive} SIP rate for general multivariate functionals. Additionally, similar to \cite{dong2019new}, we remove the assumption of a 1-step minorization. Under the same assumptions, \cite{dong2019new,zhu2020asymptotic} obtain the rate with $\beta = \max\{ 1/{(2+\delta)}, 1/p, 1/4 \}$, hence our rates are tighter.
	\end{remark}

	For practical application to MCMC, it is important to assess when and for what values of $p$ is $\E_{Q}(\tau_{1}^{p}) < \infty $. For a 1-step minorization, \cite{hobert2002applicability} show that $\tau_1$ has a moment-generating function when $\{X_t\}_{t\ge 1}$ is geometrically ergodic. The next two lemmas are critical to obtaining moment conditions over $\tau$ for polynomially and geometrically ergodic Markov chains. Additionally, Lemma~\ref{lm:sum_exp} and Lemma~\ref{lm:geom_f} aid in proving the moment existence of regenerative sums. Proofs of the following lemmas are in Appendix~\ref{sec:theorem2}.
 
    \begin{lemma}\label{lm:poly}
		Let $\{X_t\}_{t\geq 1}$ be a $\pi$-stationary polynomially ergodic Markov chain of order $\xi > (2+\delta)(1 + (2 + \delta)/\delta^{*})$ for some $\delta > 0$ and $\delta^{*} > 0$. Then for all $p \in (0, \xi)$, $\E_{Q} [\tau_{1}^{p}] < \infty$.
	\end{lemma}


    \begin{lemma}\label{lm:geom}
		Let $\{X_t\}_{t\geq 1}$ be a $\pi$-stationary geometrically ergodic Markov chain. Then $\E_{Q} [\tau_{1}^{p}] < \infty$ for any $p > 1$. 
	\end{lemma}


    \begin{lemma}
		\label{lm:sum_exp}
		Let $\{X_t\}_{t\geq 1}$ be a $\pi$-Harris ergodic Markov chain so that \eqref{eq:mino} holds. Let  $\E_{\pi}\left(\Vert f(X) \Vert^{p+\delta^{*}}\right) < \infty$ for some $p > 1$ and $\delta^{*}> 0$ and $\E_{\pi}(\tau_{1}^{\phi}) < \infty$ for $\phi > p(p+\delta^{*})/\delta^{*}$. Then, 
		\[
		\E_{\pi}\left[\left(\sum_{i=1}^{\tau_{1}} \Vert f(X_{i}) \Vert\right)^{p}\right] < \infty \,.
		\]
	\end{lemma}


    \begin{lemma}
		\label{lm:geom_f}
		Let $\{X_t\}_{t\geq 1}$ be a $\pi$-Harris ergodic Markov chain so that \eqref{eq:mino} holds. Further, let $\{X_n\}_{n \geq 1}$ be geometrically ergodic and $\E_{\pi}\left[\left(\Vert f(X) \Vert\right)^{p+\delta^{*}}\right] < \infty$ for some $p > 1$ and $\delta^{*} > 0$. Then
		\begin{equation*}
			\E_{\pi}\left[\left(\sum_{i=1}^{\tau_{1}} \Vert f(X_{i}) \Vert\right)^{p}\right] < \infty.
		\end{equation*}
	\end{lemma}

    
    The above lemmas allow the following result.
	\begin{theorem}
		\label{thm:mcmc_sip}
		Let $\{X_t\}_{t\geq 1}$ be a $\pi$-Harris ergodic Markov chain and thus \eqref{eq:mino} holds. Suppose the chain is either
		\begin{enumerate}[(a)]
			\item polynomially ergodic of order $\xi > (2 + \delta)(1 + (2 + \delta)/\delta^{*})$ for some $\delta > 0$ and $\delta^{*} > 0$ and $\E_{\pi}\left(\Vert f(X) \Vert^{2+\delta+\delta^{*}}\right) < \infty$ ; or,
			\item  geometrically ergodic and $\E_{\pi}\left(\Vert f(X) \Vert^{2+\delta+\delta^{*}}\right)  < \infty$ for some $\delta > 0$ and $\delta^{*} > 0$,
		\end{enumerate}	
		then \eqref{eq:main_sip} holds with $\beta = \max\{ 1/{(2+\delta)}, 1/4\}$.
	\end{theorem}

	\begin{proof}
		See Appendix~\ref{sec:theorem2}.
	\end{proof}

	\begin{remark}
		Theorem~\ref{thm:mcmc_sip} presents reasonable and verifiable conditions for the existence of an SIP. Similar conditions (slightly weaker in the case of polynomial ergodicity) are also sufficient for the existence of a CLT \citep[see][]{jones2004markov}. Theorem~\ref{thm:mcmc_sip} marks a three-fold improvement over the existing results of \cite{jones2006fixed} in MCMC; (i) the assumption of a $1$-step minorization is completely removed, (ii) the inclusion of an explicit result for polynomially ergodic Markov chains, and (iii) the critical extension to multivariate functionals. 
	\end{remark}

    \begin{remark}
        For any Harris ergodic Markov chain, an $l$-step minorization always holds for some $l \geq 1$. Hence wide-sense regeneration exists for the underlying Markov chain. Consequently the one-dependent random sequences $\{(Z_{t}, \tau_{t})\}_{t \ge 1}$ can be constructed for any $l \geq 2$. Different $l$ yields different covariance structure $\Sigma_{Z}$ and different $\mu$ in such a way that the quantity $\Sigma_{Z} / \mu$ remains same. This allows for a larger class of representations of the asymptotic covariance matrix. Detecting wide-sense regenerations in real applications can be quite challenging. Nevertheless, in Section~\ref{sec:regen} and Appendix~\ref{appendix1}, we present its theoretical framework.
        
    \end{remark}
	
	\section{Application to MCMC variance estimation}

    \subsection{Batch means estimator}
	\label{sec:practical}
	
	Having generated the process $\{X_t\}_{t=1}^{n}$ through an MCMC algorithm, the samples are employed to estimate $\E_{\pi}[f(X)]$ via the Monte Carlo average since,
	\[
	\hat{f}_n: = \frac{1}{n} \sum_{t=1}^{n} f(X_{t}) \overset{a.s.}{\rightarrow} \E_{\pi}[f(X)] \qquad  \text{as }  n \rightarrow \infty\,.
	\]
	When a Markov chain CLT  holds for $\hat{f}_n$, there exists a positive-definite matrix $\Sigma_{f}$ such that as $n \to \infty$,
	\begin{equation}
		\sqrt{n}\left( \hat{f}_n - \E_{\pi}[f(X)]\right) \overset{\text{d}}{\rightarrow} \text{N} (0, \Sigma_{f})\,, \label{eq:MCLT}
	\end{equation} 
	where 
	\[
	\Sigma_{f} = \Var_{\pi}[f(X)] + \sum_{s=1}^{\infty}  \left[\cov_{\pi}(f(X_{1}), f(X_{1+s})) + \cov_{\pi}(f(X_{1}), f(X_{1+s}))^{\top}   \right]\,.
	\]
	See \cite{jones2004markov} for sufficient conditions for a Markov chain CLT to hold. We note that Theorem~\ref{thm:mcmc_sip} also implies a Markov chain CLT and thus $\Sigma_f = \Sigma_Z/\mu$.

	Estimation of $\Sigma_f$ is widely discussed, both in the univariate case \citep{berg2022efficient,chak:khare:bhatt:2022,damerdji1991strong,geyer1992practical,damerdji1995mean,jones2006fixed,flegal2010batch}, and the multivariate case  \citep{kosorok2000monte,dai2017multivariate,seila1982multivariate,vats2018strong,vats2019multivariate,lugsail2021vats}. Estimators of $\Sigma_f$ are employed in deciding when to stop the simulation. Thus, the MCMC simulation stops at a random time and \cite{glynn1992asymptotic} show that validity of the subsequent estimators require strong consistency of estimators of $\Sigma_f$. Much effort has thus gone into ensuring that estimators of $\Sigma_f$ are strongly consistent. One particular estimator that stands out due its computational efficiency and theoretical underpinnings, is the batch-means estimator.
	
	An existence of a multivariate SIP is assumed for the strong consistency of the batch-means estimators \citep{vats2019multivariate}. \cite{vats2018strong} showed that if the Markov chain is polynomially ergodic, then for some unknown $0< \lambda < 1/2$ a multivariate SIP holds with rate $n^{1/2 - \lambda}$. Strong consistency also depends on choosing appropriate values for the tuning parameter, the batch size. These values depend on the rate $\lambda$, which is unknown, thus making it difficult to verify the conditions. Theorem~\ref{thm:mcmc_sip} overcomes this issue considerably as we will now elucidate.
	
	Let the Monte Carlo sample size $n = a_n b_n$ where $a_n$ denotes the number of batches and $b_n$ is the batch size. 
	For $k = 1, \dots, a_n$ define the mean vector of the $k^{\text{th}}$ batch as, $\bar{f}_{k} = \left(b_n^{-1}\sum_{t= (k-1)b_{n}+1}^{k b_{n}} f(X_{t})\right)$. Then the batch-means estimator of $\Sigma_f$ is
	\begin{equation*}
		\hat{\Sigma}_{\text{BM}} = \frac{b_{n}}{a_{n}-1} \sum_{k=1}^{a_{n}} (\bar{f}_{k} - \hat{f}_n)(\bar{f}_{k} - \hat{f}_n)^{\top}.
	\end{equation*}
	
	\begin{assum}
		\label{ass:bn}
		The batch size $b_n$ is such that
		\begin{enumerate}[(a)]
			\item $b_n \to \infty$ and $n/b_n \to \infty$ as $n \to \infty$ where, $b_n$ and $n/b_n$ are monotonically increasing,
			\item there exists a constant $c \geq 1$ such that $\sum_n (b_n n^{-1})^c < \infty$.
		\end{enumerate}
	\end{assum}
	
	Often $b_n = \lfloor n^{\nu} \rfloor$ for some $\nu > 0$ so that Assumption~\ref{ass:bn} is trivially satisfied. Common choices in the literature are $b_n = \lfloor n^{1/3} \rfloor$ and $b_n = \lfloor n^{1/2} \rfloor$. The following theorem from \cite{vats2019multivariate} presents the conditions for strong consistency of the batch-means estimator.
	\begin{theorem}[\cite{vats2019multivariate}]
		\label{thm:bm_cons}
		Suppose $f$ is such that $\E_{\pi}\left(\Vert f(X) \Vert^{2 + \delta}\right) < \infty$ for some $\delta > 0$ and let the Markov chain be polynomially ergodic of order $\xi > (1 + \epsilon)(1 + 2/\delta)$ for some $\epsilon > 0$. If $b_n$ satisfies Assumption~\ref{ass:bn} and $b_n^{-1}\log(n) \kappa^2(n) \to 0$ as $n\to \infty$, then $\hat{\Sigma}_{\text{BM}} \to \Sigma$  with probability 1 as $n \to \infty$. 
	\end{theorem}
	For the univariate case, \cite{jones2006fixed} assumed the same conditions on the batch size. For geometrically ergodic Markov chains under a $1$-step minorization, \cite{jones2006fixed} showed a univariate SIP holds with $\kappa(n) = n^{\beta} \log n$, where $\beta$ is the same as Theorem~\ref{thm:mcmc_sip}. This implies that the batch size should be chosen such that $\nu > \max\{2/(2 + \delta), 1/2 \}$. For geometrically ergodic Markov chains, this is the best known batch size condition. We note, as did \cite{damerdji1991strong,damerdji1995mean}, that this condition excludes common choices of $b_n$.

	For the multivariate case, the only known MCMC SIP result was that of \cite{vats2018strong} with $\kappa(n) = n^{1/2 - \lambda}$ for some unknown $0 < \lambda < 1/2$. This implies a choice of $\nu > 1 - 2\lambda$ but since $\lambda$ is unknown, the condition cannot be verified. On the other hand, using Theorem~\ref{thm:mcmc_sip}, we obtain the condition that $\nu > \max\{2/(2 + \delta), 1/2 \}$. This is the same condition as the univariate case, but now only requiring polynomial ergodicity.

    \subsection{Regenerative estimator}
	\label{sec:reg_est}
	Regenerations, particularly wide-sense regenerations, are notoriously difficult to identify. Nevertheless, it is natural to consider a wide-sense regenerative estimator of $\E_{\pi}[f(X)]$ and $\Sigma_{f}$, even if it is for theoretical exposition. For $1$-step minorization, \cite{hobert2002applicability,seila1982multivariate} present regenerative estimators, and for univariate wide-sense regenerative processes, \cite{henderson2001regenerative} provide an estimator of $\E_{\pi}[f(X)]$. 
	
	Denote the number of regenerations by $R$. By a strong law of large numbers for $1$-dependent processes, $T_R/R \overset{a.s.}{\to} \mu$ as $R \to \infty$. Further as $R \to \infty$, by Lemma~\ref{lm:Ef_eta-mu} in Appendix~\ref{app:pre_lemma},
	\begin{equation}
		\label{eq:regen_est}
		\tilde{f}_R:= \dfrac{1}{T_R} \sum_{t=1}^{T_{R}} f(X_{t}) = \dfrac{ \sum_{j=1}^{R} Z_j}{\sum_{j=1}^{R} \tau_j} \overset{a.s.}{\to} \dfrac{\eta}{\mu} = \E_{\pi}[f(X)] \,.
	\end{equation}
	By using Slutsky's theorem,
	\begin{align*}
		\sqrt{R}\left(\tilde{f}_R - \E_{\pi}[f(X)] \right)  &= \sqrt{R}\left(\dfrac{ \sum_{j=1}^{R} Z_j}{  \sum_{j=1}^{R} \tau_j} - \E_{\pi}[f(X)] \right)\\
		& =  \dfrac{R}{T_R} \dfrac{1}{\sqrt{R}} \left(\sum_{j=1}^{R} Z_j - T_R \E_{\pi}[f(X)]  \right)\\ 
		& = \frac{R}{T_{R}} \frac{1}{\sqrt{R}} \sum_{j=1}^{R}\left(Z_{j} - \tau_{j}\E_{\pi}[f(X)] \right) \\
		& \overset{d}{\to} N_d(0, \Sigma_f/\mu)\,.\numberthis \label{eq:Z_j_partial_sum}
	\end{align*}
	Thus, if regenerations can be identified, the regenerative estimator in \eqref{eq:regen_est} can be used in place of $\hat{f}_n$. Regenerative estimators of $\E_{\pi}[f(X)]$ from $1$-step minorization have been employed by \cite{mykland1995regeneration,roy2007convergence,chen2018mcmc}. 
 
    For an observed Markov chain of length $n$, the estimation of the asymptotic variance $\Sigma_{f}$ in \eqref{eq:Z_j_partial_sum} is  done by estimating $\Sigma_{Z}$ as described in \eqref{eq:sigma_z} and $\mu$ separately. Define $\bar{Z} := R^{-1}\sum_{i=1}^{R} Z_{i}$. We can estimate $\mu$ and $\Sigma_{Z}$ with,
    \begin{align*}
    	  \hat{\mu} &:= \frac{1}{R} \sum_{i=1}^{R} \tau_{i} \,, \qquad{\text{and}} \numberthis \label{est:mu} \\ 
	      \hat{\Sigma}_{Z} & := \frac{1}{R} \sum_{i=1}^{R} (Z_{i} - \bar{Z})(Z_{i} - \bar{Z})^{T} + \frac{1}{R} \sum_{i=1}^{R - 1} (Z_{i} - \bar{Z})(Z_{i+1} - \bar{Z})^{T} \\
    & \qquad + \frac{1}{R} \sum_{i=1}^{R - 1} (Z_{i+1} - \bar{Z})(Z_{i} - \bar{Z})^{T} \numberthis \label{est:Sigma_Z}
    \end{align*}
    
    Using \eqref{est:mu} and \eqref{est:Sigma_Z}, the estimator for $\Sigma_{f}$ is $\hat{\Sigma}_{f} = \hat{\Sigma}_{Z} / \hat{\mu}$.

    Employing these estimators in practice has two significant challenges: (i) detecting wide-sense regenerations can be difficult since one requires knowledge of $l$ and the minorization kernel $Q$, and (ii) when minorizations can be established, the bounds in the minorization constant are fairly weak, yielding prohibitively large regeneration times. In Appendix~\ref{appendix1},  we present a general technique of identifying wide-sense regenerations for random scan Gibbs samplers via the Gibbs sampler of \cite{AC1993probit}. As it turns out, the practical implementation fails the theoretical framework since close to no regenerations are identified even in long runs of the chain. This is predominantly due to the weak bound obtained in the minorization constant. Given the framework presented in this section, it would be worthwhile pursuing improvements on this minorization constant in the future.

	\section{Acknowledgments} 
	The authors are grateful to Prof.~Galin Jones, Prof.~James Flegal, Prof.~Jing Dong and Sanket Agrawal for helpful conversations. Dootika Vats is supported by SERB (SPG/2021/001322).

	\appendix
	\section{Some preliminary results}
	\label{app:pre_lemma}
	
	\begin{lemma}
		\label{lm:Ef_eta-mu} Let $\{X_{t}\}_{t\ge 1}$ be a $\pi$-Harris ergodic Markov chain. Recall $f, \eta$ and $\mu$ from Section~\ref{sec:regen}. Then, 
		\[
		\E_{\pi}[f(X)] = \frac{\eta}{\mu}\,.
		\]
	\end{lemma}
	
	\begin{proof}
		\label{pf:lemma-1} By \cite{glynn2011wide}, $\{(Z_{k}, \tau_{k}): k \ge 1\}$ form a stationary $1$-dependent process. By a strong law of large numbers for $1$-dependent processes,
		\begin{align}
			\frac{1}{R} \sum_{i=1}^{R} Z_{i} & \overset{a.s.}{\rightarrow} \E_{Q} (Z_{1}) \text{        as      }R \rightarrow \infty \ \ \text{and,} \label{eq:cr1} \\
			\frac{1}{R} \sum_{i=1}^{R} \tau_{i} & \overset{a.s.}{\rightarrow} \E_{Q} (\tau_{1}) \text{        as      }R \rightarrow \infty \label{eq:cr2}.
		\end{align}
		By a strong law for ergodic Markov chains and from \eqref{eq:cr2},
		\begin{align*}
			& \frac{1}{T_{R}} \sum_{t=1}^{T_{R}} f(X_{t}) = \frac{1/R \sum_{i=1}^{R} Z_{i}}{1/R \sum_{i=1}^{R} \tau_{i}} \overset{a.s.}{\rightarrow} \E_{\pi}[f(X)]  \text{        as      }R \rightarrow \infty \\
			\Rightarrow & \left(1/R \sum_{i=1}^{R} \tau_{i} \right) \frac{1/R \sum_{i=1}^{R} Z_{i}}{1/R \sum_{i=1}^{R} \tau_{i}} \overset{a.s.}{\rightarrow} \E_{Q}(\tau_{1}) \E_{\pi}[f(X)]  \text{        as      }R \rightarrow \infty\\
			\Rightarrow & \left(1/R \sum_{i=1}^{R} Z_{i} \right) \overset{a.s.}{\rightarrow} \E_{Q}(\tau_{1}) \E_{\pi}[f(X)]  \text{        as      }R \rightarrow \infty \numberthis \label{eq:cr3}.
		\end{align*}
		Thus, by \eqref{eq:cr1} and \eqref{eq:cr3},
		\begin{equation}
			\E_{Q} (Z_{1}) = \E_{Q}(\tau_{1}) \E_{\pi}[f(X)]  \Rightarrow \E_{\pi}[f(X)] = \frac{\eta}{\mu}  \label{expectation}.
		\end{equation}
	\end{proof}
	
	The following lemma will be employed for the proof of Theorem~\ref{thm:mcmc_sip} and is an extension of \citet[Lemma 1]{hobert2002applicability} to the multivariate and  $l$-step minorization case.
	\begin{lemma}
		\label{lm:hobertext}
		Let $\{X_t\}_{t\geq 1}$ be a $\pi$-Harris ergodic Markov chain so that \eqref{eq:mino} holds.  Then, for any measurable function $\Psi: \mathcal{X}^{\infty} \rightarrow \mathbb{R}^{d}$
		\begin{equation}
			\E_{\pi}\Vert \Psi(X_{1}, X_{2}, X_{3}, \ldots)\Vert \geq \E_{\pi}[h(X)] \E_{Q}\Vert \Psi(X_{1}, X_{2}, X_{3}, \ldots)\Vert\,.
		\end{equation}
	\end{lemma}
	
	\begin{proof}
		For all $A \in \mathcal{B}(\mathcal{X})$ and $x \in \mathcal{X}$,
		\begin{equation}
			\pi(A) = \left(\pi P^{l} \right)(A) = \int_{\mathcal{X}} \pi(dx) P^{l}(x, A) \ge  Q(A) \int_{\mathcal{X}} h(x) \pi(x) = Q(A) \E_{\pi}[h(X)]. \label{eq:pi_Q_inq}
		\end{equation}
		By taking conditional expectation over $X_{1}$
		\begin{equation}
			\E_{\pi}\Vert \Psi(X_{1}, X_{2}, X_{3}, \ldots)\Vert = \E_{\pi}\left(\E\{\Vert \Psi(X_{1}, X_{2}, X_{3}, \ldots)\Vert \mid X_{1}\}\right) \label{eq:conditional}.
		\end{equation}
		
		Since $\E\{\Vert \Psi(X_{1}, X_{2}, X_{3}, \ldots)\Vert \mid X_{1}\}$ is a positive function of $X_{1}$, by \eqref{eq:pi_Q_inq} and \eqref{eq:conditional}
		\begin{align*}
			\E_{\pi}\Vert \Psi(X_{1}, X_{2}, X_{3}, \ldots)\Vert & = \int_{\mathcal{X}} \E\{\Vert \Psi(X_{1}, X_{2}, X_{3}, \ldots)\Vert \mid X_{1} = x\} \pi(dx) \\ 
			& \ge \E_{\pi}[h(X)] \int_{\mathcal{X}} \E\left(\Vert \Psi(X_{1}, \ldots)\Vert \mid X_{1} = x\right) Q(dx) \\
			&  = \E_{\pi}[h(X)] \E_{Q}\Vert \Psi(X_{1}, \ldots)\Vert\\
			\Rightarrow \E_{\pi}\Vert \Psi(X_{1}, \ldots)\Vert & \ge  \E_{\pi}[h(X)] \E_{Q}\Vert \Psi(X_{1}, \ldots)\Vert.
		\end{align*}
	\end{proof}
	
	\section{Proof of Theorem~\ref{Ths}}
	\begin{proof}
		\label{Thp} We will start by showing some moment properties of the sequence of regeneration times in order to prove strong convergence. Denote the number of regenerations in a sample of size $n$ by: 
		\[
		\xi(n) := \sup\{k \ge 1: T_{k}\leq n\} = \inf\{k \ge 1: T_{k+1} > n \}\,.
		\]
		
		By \cite{glynn2011wide}, $\{ (Z_{k}, \tau_{k}): k \ge 1\}$ is a stationary $1$-dependent process. Further, by the conditions in the theorem, $\E_{Q}(\tau_{1}^{p}) < \infty$ for $p > 1$. Define $p' = \min\{2, p\}$ for $p > 1$. Thus, $\E_{Q}(\tau_{1}^{p'}) < \infty$. By a Marcinkiewicz-Zygmund strong law of large numbers for 1-dependent process \citep[see][for e.g.]{samur2004regularity}, for $1 < p' \leq 2$,
		\begin{equation}
			\left| \sum_{i=1}^{\xi(n)} \tau_{i} - \xi(n) \mu \right| =  \vert T_{\xi(n)} - \xi(n) \mu \vert \overset{a.s.}{=} \smallO (\xi(n)^{1/p'})\, . \label{eq:bound_regtime}	
		\end{equation}

		With $n < T_{\xi(n)+1}$,  subtracting $\xi(n)\mu$ from both side
		\begin{align*}
			n - \xi(n)\mu < T_{\xi(n)+1} - (\xi(n)+1)\mu + \mu.
        \end{align*}
        Now, using \eqref{eq:bound_regtime}, $T_{\xi(n)+1} - (\xi(n)+1)\mu \overset{a.s.}{=} \mathcal{O}((\xi(n)+1)^{1/p'})$  and hence $T_{\xi(n)+1} - (\xi(n)+1)\mu + \mu \overset{a.s.}{=} \mathcal{O}((\xi(n)+1)^{1/p'})$. Thus
        \begin{align*}
			 n - \xi(n)\mu \overset{a.s.}{=} & \mathcal{O}((\xi(n)+1)^{1/p'})\\
			\Rightarrow n - \xi(n)\mu \overset{a.s.}{=} & \mathcal{O}(n^{1/p'})
			\Rightarrow \xi(n) \overset{a.s.}{=} n / \mu + \mathcal{O}(n^{1/p'}) \numberthis \label{eq:xi_expan}\\
			\Rightarrow \vert \xi(n) - n/ \mu \vert \overset{a.s.}{=} & \mathcal{O}(n^{1/p'}) . \numberthis \label{eq-10}
		\end{align*}
		
		Further, by the assumption in \eqref{eq:assm-2} and Lemma~\ref{lm:hobertext},
		\begin{align*}
			& \E_{\pi} \left[\left( \sum_{t=1}^{\tau_{1}} \left\Vert f(X_{t}) - \frac{\eta}{\mu} \right\Vert \right)^{2+\delta}\right] < \infty \\
			\Rightarrow & \E_{Q} \left[ \left( \sum_{t=1}^{\tau_{1}} \left\Vert f(X_{t}) - \frac{\eta}{\mu} \right\Vert \right]^{2+\delta}\right] < \infty \\
			\Rightarrow & \E_{Q} \left[\left(\left\Vert Z_{1} - \tau_{1} \frac{\eta}{\mu}  \right\Vert \right)^{2+\delta}\right] < \infty  \\
			\Rightarrow & \left\Vert \E_{Q}\left(Z_{1} - \tau_{1} \frac{\eta}{\mu}\right) \right\Vert^{2+\delta} < \infty\,, \numberthis \label{eq:modify_assm2}
		\end{align*}
		where the last implication follows from Jensen's inequality. 
  
        Note that $\{(X_{t}, \delta_{t})\}_{t \ge 1}$ follows all the properties of a process generated by the alternative sampling strategy as illustrated in Section~\ref{sec:regen}. So, $X_{1} \sim Q$ and $(X_{2}, \delta_{2}) \ldots, (X_{\tau_{1}+l-1}, \delta_{\tau_{1}+l-1})$ are serially generated from the initial value. Define, $S_{k} := \left( \sum_{i=\tau_{k-1}+1}^{\tau_{k}} X_{i}, \sum_{i=\tau_{k}+1}^{\tau_{k}+l-1} X_{i}, \delta_{\tau_{k-1}+1}, \ldots, \delta_{\tau_{k-1}+l-1} \right)^{\text{T}}$. By construction $S_{k}$'s are iid vectors \cite[see][Section 2]{glynn2011wide}. Consequently, $(Z_{1}, \tau_{1}) = g(S_{1})$ for some measurable function $g(\cdot)$. Again $X_{T_{1}+l} \sim Q$ and independent to all previous elements in the chain; define $(Z_{2}, \tau_{2}) = g(S_{1}, S_{2})$. Hence, for all $k \ge 1$ we can say $(Z_{k}, \tau_{k}) = g(S_{k-1}, S_{k})$ where $S_{0} :=  0$. Also, from \eqref{eq:assm-2} it directly follows that $\E_{\pi}\left(\Vert f(X) \Vert^{2+\delta}\right) < \infty$. Thus, Condition A of \cite{liu2009strong} is satisfied. By \citet[Theorem 2.1]{liu2009strong} for the stationary $1$-dependent process $\{ (Z_{k} - \tau_{k} \frac{\eta}{\mu}): k \ge 1\}$ and $\{ W(t) : t > 0 \} $, a $d$-dimensional standard Wiener process,
		\begin{align*}
			& \left\Vert \sum_{k=1}^{\xi(n)} Z_{k} - T_{\xi(n)} \frac{\eta}{\mu} - \Sigma_{Z}^{1/2} W(\xi(n)) \right\Vert  \overset{a.s.}{=} \mathcal{O}(\xi(n)^{1/(2+\delta)}) \\
			\Rightarrow & \left\Vert \sum_{k=1}^{\xi(n)} Z_{k} - T_{\xi(n)} \frac{\eta}{\mu} - \Sigma_{Z}^{1/2} W(\xi(n)) \right\Vert  \overset{a.s.}{=} \mathcal{O}(n^{1/(2+\delta)}) \numberthis \label{eq:1st bound}.
		\end{align*}
		
		By \eqref{eq:xi_expan} and \citet[Proposition 1.2.1]{csorgo2014strong}, 
		\begin{equation*}
			\Vert W(\xi(n)) - W(n/ \mu) \Vert \overset{a.s.}{=} \mathcal{O}(b_{n})\,,
		\end{equation*}
		where for positive constants $c$ and $c'$,
		\begin{align*}
			b_{n} & = \left( 2c n^{1/p'}\left(\log \left( \frac{n/\mu}{n^{1/p'}}\right) + \log \log \left(n/\mu\right)\right)\right)^{1/2}\\
			& = \left( 2c n^{1/p'}\left(\log \left( \frac{n^{1-1/p'}}{\mu}\right) + \log \log \left(n/\mu\right)\right)\right)^{1/2}\\
			& < c' n^{1/(2p')}\left(\log n\right) \,.
		\end{align*}
		Consequently,
		\begin{equation}
			\Vert W(\xi(n)) - W(n/ \mu) \Vert \overset{a.s.}{=} \mathcal{O}(n^{1/(2p')}\log n) \label{eq:3rd bound}.
		\end{equation}

		Using triangle inequality and since $T_{\xi(n)} < n < T_{\xi(n) + 1}$,
		\begin{align*}
			Y_{\xi(n)} := & \sum_{i=T_{\xi(n)}+1}^{T_{\xi(n)+1}} \left\Vert f(X_{i}) - \frac{\eta}{\mu} \right\Vert \\
   > & \left\Vert \sum_{i= T_{\xi(n)}+1}^{n} \left( f(X_{i}) - \frac{\eta}{\mu} \right) \right\Vert\\ 
   = & \left\Vert \sum_{i=1}^{n} \left( f(X_{i})  - \frac{\eta}{\mu} \right) -  \sum_{k=1}^{\xi(n)} \left( Z_{k} - \tau_{k} \frac{\eta}{\mu} \right) \right\Vert. \numberthis \label{eq:define_Y}
		\end{align*}
		By construction, \( \{ Y_{k} \}_{k\ge 1} \) are positive and identical random variables generated from sum of absolute values of correlated units sampled through wide-sense regeneration. By  the integral transformation inequality and the assumption in equation-\eqref{eq:assm-2}
		\begin{align*}
			\sum_{i=1}^{\infty} \Pr \left({Y_{i}}^{2+\delta} > i \right) & = \sum_{i=1}^{\infty} \Pr \left({Y_{1}}^{2+\delta} > i \right) \\
			& < \int_{1}^{\infty} \Pr \left({Y_{1}}^{2+\delta} > x \right) dx \\
			& < \int_{0}^{\infty} \Pr \left({Y_{1}}^{2+\delta} > x \right) dx \\
			& = \E_{Q}[{Y_{1}}^{2+\delta}]\\
            & < \infty.
		\end{align*}
		Consequently,
		\begin{equation}
			\sum_{i=1}^{\infty} \Pr \left({Y_{i}}^{2+\delta} > i \right) < \infty
			\Rightarrow \sum_{i=1}^{\infty} \Pr \left(Y_{i} > i^{\frac{1}{2+\delta}} \right) < \infty .
		\end{equation}
		Thus by Borel-Cantelli lemma
		\begin{equation}
			Y_{n} \overset{a.s.}{=} \mathcal{O} \left(n^{\frac{1}{2+\delta}} \right) \qquad \text{     as     }  n \to \infty \label{eq:Y_tight_O}.
		\end{equation}
		
		From \eqref{eq:define_Y} and \eqref{eq:Y_tight_O} as $n \to \infty$
		\begin{align*}
			& Y_{\xi(n)} \overset{a.s.}{=} \mathcal{O} \left(\xi(n)^{\frac{1}{2+\delta}} \right) \\
			\Rightarrow & \left\Vert \sum_{i=1}^{n} \left( f(X_{i})  - \frac{\eta}{\mu} \right) -  \sum_{k=1}^{\xi(n)} \left(Z_{k} - \tau_{k} \frac{\eta}{\mu} \right) \right\Vert \overset{a.s.}{=} \mathcal{O} \left(n^{\frac{1}{2+\delta}} \right) \numberthis \label{eq:2nd bound}.
		\end{align*}
		Using the triangle inequality and \eqref{eq:1st bound}, \eqref{eq:3rd bound}, and \eqref{eq:2nd bound} 
		\begin{flalign*}
			\left\Vert \sum_{i=1}^{n}f(X_{i}) - n \frac{\eta}{\mu} -  \frac{\Sigma_{Z}^{1/2}}{\sqrt{\mu}}W(n) \right\Vert < & \left\Vert \sum_{i=1}^{n} \left( f(X_{i})  - \frac{\eta}{\mu} \right) -  \sum_{k=1}^{\xi(n)} \left(Z_{k} - \tau_{k} \frac{\eta}{\mu} \right) \right\Vert \\
            & +  \left\Vert \sum_{k=1}^{\xi(n)} Z_{k} - T_{\xi(n)} \frac{\eta}{\mu} - \Sigma_{Z}^{1/2} W(\xi(n)) \right\Vert \\ & + \left\Vert \Sigma_{Z}^{1/2}\left(W(\xi(n)) - W\left(\frac{n}{\mu}\right)\right) \right\Vert ;
        \end{flalign*}
        \begin{flalign*}
            \implies \left\Vert \sum_{i=1}^{n}f(X_{i}) - n \frac{\eta}{\mu} -  \frac{\Sigma_{Z}^{1/2}}{\sqrt{\mu}}W(n) \right\Vert = \mathcal{O}(n^{1/(2+\delta)}) & + \mathcal{O}(n^{1/(2+\delta)})\\ & + \mathcal{O}(n^{1/p} \log n)
		\end{flalign*}
  
		Thus, with $\beta = \max\{1/(2+\delta), 1/2p'\} = \max\{1/(2+\delta), 1/2p, 1/4\}$ and by lemma-\ref{lm:Ef_eta-mu} as $n \rightarrow \infty$ 
		\begin{align*}
		    \left\Vert \sum_{i=1}^{n}f(X_{i}) - n \frac{\eta}{\mu} -  \frac{\Sigma_{Z}^{1/2}}{\sqrt{\mu}}W(n) \right\Vert & \overset{a.s.}{=} \mathcal{O}(n^{\beta}\log n)\,\numberthis \label{eq:semi-final} \\
      \Rightarrow \left\Vert \sum_{i=1}^{n} f(X_{i}) - n \E_{\pi}[f(X)] - \frac{\Sigma_{Z}^{1/2}}{\sqrt{\mu}} W(n) \right\Vert & \overset{a.s.}{=} \mathcal{O}(n^{\beta}\log n) .
		\end{align*}

		
	\end{proof}

	\section{Proof of Theorem~\ref{thm:mcmc_sip}}
	\label{sec:theorem2}
	\begin{proof}[Proof of Lemma~\ref{lm:poly}] \label{pf:poly}
		$\{X_t\}_{t \ge 1}$ is a polynomially ergodic sequence of random variables of order $\xi$ for $\xi > (2+\delta)(1 + (2 + \delta)/\delta^{*})$. So in \eqref{eq:tv_dist}, $G(n) = n^{-\xi}$. Further, from \cite{jones2004markov} we have that $\alpha(n) \leq n^{-\xi}$ for $n \geq 1$. Consequently, for $p < \xi$
		\[
		\sum_{n=1}^{\infty} n^{p-1} \alpha(n) < \sum_{n=1}^{\infty} n^{p-1} n^{-\xi} < \infty.
		\] 
		By \citet[Proposition 3.1]{samur2004regularity}
		\begin{equation}
			\E_{\pi} [\tau_{1}^{p}] < \infty \ \ \text{for} \ p \in (0, \xi). \label{eq:tau-1}
		\end{equation}
		From Lemma \ref{lm:hobertext} we have $\E_{Q} [\tau_{1}^{p}] < \infty$ for $p \in (0, \xi)$.
	\end{proof}

	
	\begin{proof}[Proof of Lemma~\ref{lm:geom}]\label{pf:geom}
		Since $\{X_t\}_{t \ge 1}$ is geometrically ergodic, $G(n) = t^{n}$ for some $0 < t < 1$. From \cite{jones2004markov}, $ \alpha(n) \leq t^{n}$ for $n\ge 1$. Consequently for $p > 1$,   
		\[
		\sum_{n=1}^{\infty} n^{p-1} \alpha(n) < \sum_{n=1}^{\infty} n^{p-1} t^{n} .
		\]
		
		By a ratio test, for all $p > 1$
		\[
		\lim_{n\rightarrow\infty}\frac{ (n+1)^{p-1} t^{n+1}}{n^{p-1} t^{n}} = \lim_{n\rightarrow\infty} (1+1/n)^{p-1} t = t < 1.
		\]
		So, 
		\[
		\sum_{n=1}^{\infty} n^{p-1} \alpha(n) < \infty \quad \text{for }p >1.
		\]
		By \citet[Proposition 3.1]{samur2004regularity}
		\begin{equation}
			\E_{\pi} [\tau_{1}^{p}] < \infty \ \ \text{for} \ p > 1. \label{eq:tau-2}
		\end{equation}
		From Lemma \ref{lm:hobertext}, $\E_{Q} [\tau_{1}^{p}] < \infty$ for  $p >1.$
	\end{proof}
 
	
	\begin{proof}[Proof of Lemma~\ref{lm:sum_exp}]\label{pf:sum_exp}
		For $p > 1$ and using triangle inequality for $L^{p}$-distances on $\pi$, H\"older's inequality, Markov's inequality, and, infinite sum of $p$-series,
		\begin{align*}
			& \left(\E_{\pi}\left[\left(\sum_{i=1}^{\tau_{1}} \Vert f(X_{i}) \Vert\right)^{p}\right]\right)^{1/p}\\
			& = \left(\E_{\pi}\left[\left(\sum_{i=1}^{\infty} \mathbb{I}(i \leq \tau_{1}) \Vert f(X_{i}) \Vert\right)^{p}\right]\right)^{1/p}\\
			& \leq \sum_{i=1}^{\infty} \left(\E_{\pi}\left(\mathbb{I}(i \leq \tau_{1}) \Vert f(X_{i}) \Vert^{p} \right)\right)^{1/p} \\
			& \leq \sum_{i=1}^{\infty} \left(\left(\E_{\pi}\mathbb{I}(i \leq \tau_{1})\right)^{\delta^{*}/(p+\delta^{*})} \left(\E_{\pi}\left(\Vert f(X_{i}) \Vert^{p+\delta^{*}}\right)\right)^{p/(p+\delta^{*})}\right)^{1/p} \\
			& = \left(\E_{\pi}\left(\Vert f(X) \Vert^{p+\delta^{*}}\right)\right)^{1/(p+\delta^{*})} \sum_{i=1}^{\infty} \left({\Pr}_{\pi}(\tau_{1} \ge i)\right)^{\delta^{*}/{p(p+\delta^{*})}}\\
			& \leq \left(\E_{\pi}\left(\Vert f(X) \Vert^{p+\delta^{*}}\right)\right)^{1/(p+\delta^{*})} \sum_{i=1}^{\infty} \left({\Pr}_{\pi}(\tau_{1}^{\phi} \ge i^{\phi})\right)^{\delta^{*}/{p(p+\delta^{*})}} \\
			& \leq \left(\E_{\pi}\left(\Vert f(X) \Vert^{p+\delta^{*}}\right)\right)^{1/(p+\delta^{*})} \sum_{i=1}^{\infty} \left(\frac{\E_{\pi}(\tau_{1}^{\phi})}{i^{\phi}}\right)^{\delta^{*}/{p(p+\delta^{*})}} \\
			& = \left(\E_{\pi}\left(\Vert f(X) \Vert^{p+\delta^{*}}\right)\right)^{1/(p+\delta^{*})} \left(\E_{\pi}(\tau_{1}^{\phi})\right)^{\delta^{*}/{p(p+\delta^{*})}} \sum_{i=1}^{\infty} \left(\frac{1}{i^{\phi}}\right)^{\delta^{*}/{p(p+\delta^{*})}} < \infty \numberthis \label{eq:l2}.
		\end{align*}
	\end{proof}

 
	\begin{proof}[Proof of Lemma~\ref{lm:geom_f}]\label{pf:geom_f}
		Since $\{X_t\}_{t\geq 1}$ is geometrically ergodic, by \eqref{eq:tau-2}, $\E_{\pi}[\tau_{1}^q] < \infty$  for $q > 1$. Proceeding similarly as Lemma~\ref{lm:sum_exp}, $\E_{\pi}\left(\left(\sum_{i=1}^{\tau_{1}} \Vert f(X_{i}) \Vert\right)^{p}\right) < \infty$.
	\end{proof}

    \label{pf-th2}	
	\begin{proof}[Proof of Theorem~\ref{thm:mcmc_sip}]
		$(a)$ $\{X_t\}_{t \ge 1}$ is polynomially ergodic  of order $\xi$ for $\xi > (2+\delta)(1 + (2 + \delta)/\delta^{*})$. Thus \eqref{eq:tau-1} holds. Then by Lemma \ref{lm:poly}, 
		\begin{equation}
			\E_{Q} [\tau_{1}^{p}] < \infty \qquad \text{for } p \in (0, \xi). \label{tau1}
		\end{equation}  
		By the assumption in the theorem, $\E_{\pi}\left(\Vert f(X) \Vert^{2+\delta+\delta^{*}}\right) < \infty$ for some $\delta > 0$ and $\delta^{*} > 0$. By \eqref{eq:tau-1} and Lemma~\ref{lm:sum_exp}
		\begin{align*}
			& \E_{\pi}\left[\left(\sum_{i=1}^{\tau_{1}} \Vert f(X_{i}) \Vert\right)^{2+\delta}\right] < \infty 
			\Rightarrow \E_{\pi}\left[\left(\sum_{i=1}^{\tau_{1}} \left\Vert f(X_{i}) - \frac{\eta}{\mu} \right\Vert \right)^{2+\delta}\right] < \infty. \numberthis \label{f1}
		\end{align*}
		
		By Theorem~\ref{Ths}, \eqref{tau1}, and \eqref{f1},
		\begin{equation}
			\left\Vert \sum_{i=1}^{n} f(X_{i}) - n \E_{\pi}[f(X)] -  \frac{\Sigma_{Z}^{1/2}}{\sqrt{\mu}} W(n) \right\Vert \overset{a.s.}{=} \mathcal{O}(n^{\beta}\log(n))
		\end{equation}
		as $n \rightarrow \infty$ where $\beta = \max \{1/(2+\delta), 1/(2 p), 1/4\} \ \ \forall p \in (0, \xi)$. Since $\xi > 2$ always, $\beta = \max\{1/(2+\delta), 1/4\}$.
		
		\medskip
		$(b)$ Let  $\{X_t\}_{t\ge 1}$ be geometrically ergodic. So \eqref{eq:tau-2} holds. By Lemma~\ref{lm:geom}
		\begin{equation}
			\E_{Q} [\tau_{1}^{p}] < \infty \ \forall p >1. \label{allmom}
		\end{equation}
		
		Since $\E_{\pi}\left(\Vert f(X) \Vert^{2+ \delta+ \delta^{*}}\right) < \infty$ for some $\delta > 0$ and $\delta^{*} > 0$, by \eqref{eq:tau-2} and Lemma~\ref{lm:geom_f}
		\begin{align*}
			& \E_{\pi}\left[\left(\sum_{i=1}^{\tau_{1}} \Vert f(X_{i}) \Vert\right)^{2+\delta}\right] < \infty 
			\Rightarrow \E_{\pi}\left[\left(\sum_{i=1}^{\tau_{1}} \left\Vert f(X_{i}) - \frac{\eta}{\mu} \right\Vert \right)^{2+\delta}\right] < \infty. \numberthis \label{f2}
		\end{align*}
		
		By Theorem~\ref{Ths}, \eqref{allmom}, and \eqref{f2}, with $\beta = \max\{1/(2+\delta), 1/4\}$ as $n \rightarrow \infty$
		\begin{equation}
			\left\Vert \sum_{i=1}^{n} f(X_{i}) - n \E_{\pi}[f(X)] -  \frac{\Sigma_{Z}^{1/2}}{\sqrt{\mu}} W(n) \right\Vert \overset{a.s.}{=} \mathcal{O}(n^{\beta}\log(n)).
		\end{equation}
	\end{proof}

    \section{Identifying wide-sense regenerations}
    \label{appendix1}
    Establishing a $1$-step minorization with the corresponding small set calculations has been done for deterministic scan Gibbs samplers (see \cite{mykland1995regeneration}, \cite{roy2007convergence}). Using a random scan version of the Gibbs sampler of \cite{AC1993probit}, we present a framework for identifying wide-sense regenerations.
    
    The setup is similar as purported in \cite{roy2007convergence}. For $i = 1, 2, \dots, n$, consider $Y_i \overset{\text{ind}}{\sim}  \text{Bernoulli} (\Phi(x_i^T\beta))$ where $x_i \in \mathbb{R}^p$ are given and $\beta \in \mathbb{R}^p$ is the vector of coefficients. The resulting likelihood is:
\begin{equation*}
    \Pr(Y_{1} = y_{1}, \cdots, Y_{n} = y_{1} \mid \beta) = \prod_{i=1}^{n} \{\Phi(x_{i}^{T}\beta)\}^{y_{i}}\{1 -\Phi(x_{i}^{T}\beta)\}^{1 - y_{i}}.
\end{equation*}
We consider a Bayesian model with a flat prior for $\beta$ for which the posterior reduces to 
\begin{align*}
    \pi(\beta \mid \mathbf{y}) & \propto \Pr(Y_{1} = y_{1}, \cdots, Y_{n} = y_{1} \mid \beta) \ \pi(\beta) \\
    & = \prod_{i=1}^{n} \{\Phi(x_{i}^{T}\beta)\}^{y_{i}}\{1 -\Phi(x_{i}^{T}\beta)\}^{1 - y_{i}}.
\end{align*}

 Although the analytical form of $\pi(\beta \mid \mathbf{y})$ is not easily achievable, \cite{AC1993probit} has provide a deterministic Gibbs sampler, which we refer to as the AC sampler. One iteration of the Markov chain is the following:
\begin{itemize}
    \item draw $\mathbf{z} = (z_{1}, z_{2}, \cdots, z_{n})^{T}$ such that ${z}_{i} \overset{iid}{\sim} \text{Truncated Normal}(X^{T}\beta_{i-1}, 1, y_{i})$ for all $i = 1, 2, \cdots, n$. If $y_{i} = 0$, $(-\infty, 0]$ will be the truncation range and if $y_{i} = 1$, $(0, \infty)$ will be the truncation range.
    \item Draw $\beta \sim \text{N}_{p}\left((X^{T}X)^{-1}X^{T}\mathbf{z}, (X^{T}X)^{-1}\right)$.
\end{itemize}

For the  $i^{\text{th}}$ iteration the deterministic scan Markov transition kernel is 
\begin{align*}
    k_{DS}(\beta_{i+1}, \textbf{z}_{i+1} \mid \beta_{i}, \textbf{z}_{i}) = \pi(\textbf{z}_{i+1} \mid \beta_{i+1}, \mathbf{y}) \pi(\beta_{i+1} \mid \textbf{z}_{i}, \mathbf{y})
\end{align*}

 The deterministic scan AC sampler is geometrically ergodic \cite[see][Theorem 1]{roy2007convergence} and regenerations can be identified using ``distinguished point" technique  of \cite{mykland1995regeneration}. For a distinguished point $\textbf{z}^{*} \in \mathbb{R}^{n}$ and a rectangular small set $D^{*} = [c_{1}, d_{1}]\times[c_{2}, d_{2}]\times\ldots[c_{p}, d_{p}]$ the minorization kernel is 

 \begin{equation*}
     Q(\beta_{i+1}, \textbf{z}_{i+1}) := \frac{1}{\epsilon} \pi(\beta_{i+1} \mid \textbf{z}^{*}, \mathbf{y}) \pi(\textbf{z}_{i+1} \mid \beta_{i+1}, \mathbf{y}) I_{D^{*}}(\beta_{i+1});
 \end{equation*}
where $\epsilon = \int_{\mathbb{R}^{p}} \int_{\mathbb{R}^{n}} \pi(\beta_{i+1} \mid \textbf{z}^{*}, \mathbf{y}) \ \pi(\textbf{z}_{i+1} \mid \beta_{i+1}, \mathbf{y}) I_{D^{*}}(\beta_{i+1}) \ d\textbf{z}_{i+1} d\beta_{i+1}$. For $t_{j}$ being the $j^{\text{th}}$ term of $t = (\textbf{z}_{i} - \textbf{z}_{*})^{T} X$, the minorization constant will be
\begin{equation*}
    s(\textbf{z}_{i}) := \frac{\epsilon \ \exp\left( \sum_{j=1}^{p}\left(c_{j} t_{j} I_{\mathbb{R}^{+}}(t_{j}) + d_{j} t_{j} I_{\mathbb{R}^{-}}(t_{j})\right)\right)}{\exp(0.5 (\textbf{z}_{i})^{T} X^{T} (X^{T} X)^{-1} X (\textbf{z}_{i}) - 0.5 (\textbf{z}^{*})^{T} X^{T} (X^{T} X)^{-1} X (\textbf{z}^{*}))}.
\end{equation*}

Consequently, the minorization holds with 
 \begin{equation}
     k_{DS}(\beta_{i+1}, \textbf{z}_{i+1} \mid \beta_{i}, \textbf{z}_{i}) \ge s(\textbf{z}_{i}) \ Q(\beta_{i+1}, \textbf{z}_{i+1}).\label{eq:mino-1-DS}
 \end{equation}

For $\Delta$ being the Dirac measure, the one-step random scan AC-sampler kernel is defined as
 \begin{align*}
     k_{RS}(\beta_{i+1}, \textbf{z}_{i+1} \mid \beta_{i}, \textbf{z}_{i}) = p \pi\left(\beta_{i+1} \mid \textbf{z}_{i}\right) \Delta_{\textbf{z}_{i}}(z_{i+1}) + (1 - p) \pi\left(z_{i+1} \mid \beta_{i+1}\right) \Delta_{\beta_{i}}(\beta_{i+1}).
\end{align*}

Similarly, the $2$-step random scan Markov transition kernel is 
\begin{align*}
    & k^{2}_{RS}(\beta_{i+2}, \textbf{z}_{i+2} \mid \beta_{i}, \textbf{z}_{i}) \\ =  & \int_{\mathbb{R}^{p}} \int_{\mathbb{R}^{n}} k_{RS}(\beta_{i+2}, \textbf{z}_{i+2} \mid \beta_{i+1}, \textbf{z}_{i+1}) k_{RS}(\beta_{i+1}, \textbf{z}_{i+1} \mid \beta_{i}, \textbf{z}_{i}) \  d\textbf{z}_{i+1} d\beta_{i+1} \\
    = & \int_{\mathbb{R}^{p}} \int_{\mathbb{R}^{n}}  \left( p \pi\left(\beta_{i+2} \mid \textbf{z}_{i+1}\right) \Delta_{\textbf{z}_{i+1}}(\textbf{z}_{i+2}) + (1 - p) \pi\left(\textbf{z}_{i+2} \mid \beta_{i+1}\right) \Delta_{\beta_{i+1}}(\beta_{i+2})\right) \\ &  \ \ \ \left(p \pi\left(\beta_{i+1} \mid \textbf{z}_{i}\right) \Delta_{\textbf{z}_{i}}(\textbf{z}_{i+1}) + (1 - p) \pi\left(\textbf{z}_{i+1} \mid \beta_{i}\right) \Delta_{\beta_{i}}(\beta_{i+1})\right) \ d\textbf{z}_{i+1} d\beta_{i+1}\\
    = & \int_{\mathbb{R}^{p}} \int_{\mathbb{R}^{n}} p^{2} \pi\left(\beta_{i+2} \mid \textbf{z}_{i+1}\right) \pi\left(\beta_{i+1} \mid \textbf{z}_{i}\right) \Delta_{\textbf{z}_{i}}(\textbf{z}_{i+1}) \Delta_{\textbf{z}_{i+1}}(\textbf{z}_{i+2}) \ d\textbf{z}_{i+1} d\beta_{i+1} \\   & \ \  + \int_{\mathbb{R}^{p}} \int_{\mathbb{R}^{n}} p(1-p) \pi\left(\beta_{i+2} \mid \textbf{z}_{i+1}\right) \pi\left(\textbf{z}_{i+1} \mid \beta_{i}\right) \Delta_{\beta_{i}}(\beta_{i+1}) \Delta_{\textbf{z}_{i+1}}(\textbf{z}_{i+2}) \ d\textbf{z}_{i+1} d\beta_{i+1} \\ &  \ \ + \int_{\mathbb{R}^{p}} \int_{\mathbb{R}^{n}} p(1 - p) \pi\left(\textbf{z}_{i+2} \mid \beta_{i+1}\right) \pi\left(\beta_{i+1} \mid \textbf{z}_{i}\right) \Delta_{\textbf{z}_{i}}(\textbf{z}_{i+1}) \Delta_{\beta_{i+1}}(\beta_{i+2}) \ d\textbf{z}_{i+1} d\beta_{i+1}\\ & \ \  + \int_{\mathbb{R}^{p}} \int_{\mathbb{R}^{n}} (1 - p)^{2} \pi\left(\textbf{z}_{i+2} \mid \beta_{i+1}\right) \pi\left(\textbf{z}_{i+1} \mid \beta_{i}\right) \Delta_{\beta_{i}}(\beta_{i+1}) \Delta_{\beta_{i+1}}(\beta_{i+2}) \ d\textbf{z}_{i+1} d\beta_{i+1}.
\end{align*}
So,
\begin{align*}
    k^{2}_{RS}\left(\beta_{i+2}, \textbf{z}_{i+2} \mid \beta_{i}, \textbf{z}_{i}\right) = & \ p^{2} \pi\left( \beta_{i+2} \mid \textbf{z}_{i+2} \right) \Delta_{\textbf{z}_{i}}(\textbf{z}_{i+2}) \\
    & + p (1 - p) \pi\left(\beta_{i+2} \mid \textbf{z}_{i+2}\right) \ \pi\left(\textbf{z}_{i+2} \mid \beta_{i}\right) \\
    & + p(1 - p) \pi\left(\textbf{z}_{i+2} \mid \beta_{i+2}\right) \pi\left( \beta_{i+2} \mid \textbf{z}_{i}\right) \\
    & + (1 - p)^{2} \pi\left(\textbf{z}_{i+2} \mid \beta_{i}\right) \Delta_{\beta_{i}}(\beta_{i+2}).\numberthis \label{eq:2stepRS}
\end{align*}
Using the 1-step minorization of the deterministic scan Gibbs sampler and with
$s'(\textbf{z}_{i}) := p(1 - p) s(\textbf{z}_{i})$, we get
\begin{align*}
 & k^{2}_{RS}\left(\beta_{i+2}, \textbf{z}_{i+2} \mid \beta_{i}, \textbf{z}_{i}\right) \\ = & \ p^{2} \pi\left( \beta_{i+2} \mid \textbf{z}_{i+2} \right) \Delta_{\textbf{z}_{i}}(\textbf{z}_{i+2}) + p (1 - p) \pi\left(\beta_{i+2} \mid \textbf{z}_{i+2}\right) \ \pi\left(\textbf{z}_{i+2} \mid \beta_{i}\right) \\
 & + p(1 - p) \pi\left(\textbf{z}_{i+2} \mid \beta_{i+2}\right) \pi\left( \beta_{i+2} \mid \textbf{z}_{i}\right) + (1 - p)^{2} \pi\left(\textbf{z}_{i+2} \mid \beta_{i}\right) \Delta_{\beta_{i}}(\beta_{i+2}) \\
\ge & \ p(1 - p) \pi\left(\textbf{z}_{i+2} \mid \beta_{i+2}\right) \pi\left( \beta_{i+2} \mid \textbf{z}_{i}\right) \\
\ge & \ p(1 - p) s(\textbf{z}_{i}) \ Q(\beta_{i+2}, \textbf{z}_{i+2}) \\
= & \ s'(\textbf{z}_{i}) \ Q(\beta_{i+2}, \textbf{z}_{i+2}) \numberthis \label{eq:mino-2-RS}.
\end{align*}

Following the ideas in \cite{mykland1995regeneration} and \cite{roy2007convergence}, we can obtain the probability of a regeneration from the observed chain in the following way:
\begin{align*}
    \eta_{i} & = \Pr \left(\delta_{i} = 1 \mid (\beta_{i}, \textbf{z}_{i}), (\beta_{i+2}, \textbf{z}_{i+2}) \right) ; \\
    & = \frac{\Pr (X_{i+2} = y \mid X_{i} = x, \delta_{i} = 1) \Pr (\delta_{i} = 1 \mid X_{i} = x) \Pr(X_{i} = x)}{\Pr (X_{i+2} = y \mid X_{i} = x) \Pr(X_{i} = x)}; \\
    & = \frac{\Pr (X_{i+2} = y \mid X_{i} = x, \delta_{i} = 1) \Pr (\delta_{i} = 1 \mid X_{i} = x)}{\Pr (X_{i+2} = y \mid X_{i} = x)}; \\
    & = \frac{s'(\textbf{z}_{i}) \ Q(\beta_{i+2}, \textbf{z}_{i+2})}{k^{2}_{RS}\left(\beta_{i+2}, \textbf{z}_{i+2} \mid \beta_{i}, \textbf{z}_{i}\right)};\\
    & = p(1 - p) \frac{\exp( - 0.5 (\textbf{z}_{i})^{T} X^{T} (X^{T} X)^{-1} X (\textbf{z}_{i})))}{\exp( - 0.5 (\textbf{z}_{*})^{T} X^{T} (X^{T} X)^{-1} X (\textbf{z}_{*})))}; \\
    & \times \exp\left\{ \sum_{i=1}^{p}(c_{i} t_{i} I_{\mathbb{R}^{+}}(t_{i}) + d_{i} t_{i} I_{\mathbb{R}^{-}}(t_{i})) \right\}; \\
    & \times \pi(\beta_{i+2} \mid \textbf{z}^{*}, \mathbf{y}) \ \pi(\textbf{z}_{i+2} \mid \beta_{i+2}, \mathbf{y}) I_{D^{*}}(\beta_{i+2});\\
    & \times \frac{1}{k^{2}_{RS}\left(\beta_{i+2}, \textbf{z}_{i+2} \mid \beta_{i}, \textbf{z}_{i}\right)}.
\end{align*}
The $\eta_{i}$ can be calculated analytically from a run of the random scan Gibbs sampler. By drawing Bernoulli samples with success probability $\eta_{i}$ we identify the regenerations if the outcome of a trial is $1$. 



For output analysis of the observed samples, the regenerative estimator of the asymptotic variance can be exploited as discussed in Section~\ref{sec:reg_est}. However, the bounds obtained in the minorization are weak enough that $\eta_i$s are prohibitively small, yielding close to no regenerations in large simulation lengths. 


	\bibliographystyle{apalike}
 
	\bibliography{mybibfile1.bib}

\end{document}